\documentclass[reqno,12pt]{amsart}

\usepackage{amsmath}
\usepackage{amssymb}
\usepackage{amsfonts}
\usepackage{amsthm}
\usepackage[foot]{amsaddr}
\usepackage{mathtools}
\usepackage{bbm}
\mathtoolsset{%
}
\usepackage{relsize}

\usepackage[utf8]{inputenc}
\usepackage[T1]{fontenc}

\usepackage{libertine}
\usepackage[libertine,libaltvw,cmintegrals
]{newtxmath}

\usepackage{dsfont}

\usepackage{mathrsfs}

\usepackage[%
cal=cm,
]
{mathalfa}

\usepackage{accents}

\usepackage[dvipsnames,svgnames]{xcolor}
\colorlet{MyBlue}{DodgerBlue!60!Black}
\colorlet{MyGreen}{DarkGreen!85!Black}

\usepackage{fullpage}
\newcommand{\afterhead}{.}

\usepackage[font=small,labelfont=bf]{caption}
\usepackage{subfigure}
\usepackage{tikz}
\usetikzlibrary{calc,patterns}
\usetikzlibrary{arrows,shapes,positioning}

\usepackage{acronym}
\usepackage{booktabs}       
\usepackage{latexsym}
\usepackage{paralist}
\usepackage{wasysym}
\usepackage{xspace}
\usepackage{comment}
\usepackage{longtable}
\usepackage[shortlabels]{enumitem}

\usepackage[authoryear,comma]{natbib}

\usepackage[pagebackref]{hyperref}
\hypersetup{
colorlinks=true,
linktocpage=true,
pdfstartview=FitH,
breaklinks=true,
pdfpagemode=UseNone,
pageanchor=true,
pdfpagemode=UseOutlines,
plainpages=false,
bookmarksnumbered,
bookmarksopen=false,
bookmarksopenlevel=1,
hypertexnames=true,
pdfhighlight=/O,
urlcolor=MyBlue!60!black,linkcolor=MyBlue!70!black,citecolor=DarkGreen!70!black, 
pdftitle={},
pdfauthor={},
pdfsubject={},
pdfkeywords={},
pdfcreator={pdfLaTeX},
pdfproducer={LaTeX with hyperref}
}

\def\EMAIL#1{\email{\href{mailto:#1}{\texttt{\upshape #1}}}}

\numberwithin{equation}{section}  
\usepackage[sort&compress,capitalize,nameinlink]{cleveref}
\crefname{app}{Appendix}{Appendices}
\crefname{assumption}{Assumption}{Assumptions}

\crefrangeformat{equation}{\upshape(#3#1#4)\textendas\testh(#5#2#6)}

\DeclarePairedDelimiter{\braces}{\{}{\}}

\DeclarePairedDelimiter{\parens}{(}{)}
\DeclarePairedDelimiter{\abs}{\lvert}{\rvert}

\DeclarePairedDelimiterX{\braket}[2]{\langle}{\rangle}{#1,#2}
\DeclarePairedDelimiterX{\inner}[2]{\langle}{\rangle}{#1,#2}
\DeclarePairedDelimiterX{\setdef}[2]{\{}{\}}{#1:#2}
\DeclarePairedDelimiterXPP{\probof}[1]{\prob}{(}{)}{}{%
#1}
\DeclarePairedDelimiterXPP{\exof}[1]{\ex}{[}{]}{}{%
#1}

\usepackage[textwidth=30mm]{todonotes}

\newcommand{\debug}[1]{#1}

\theoremstyle{plain}
\newtheorem{theorem}{Theorem}

\newtheorem*{corollary*}{Corollary}
\newtheorem{lemma}[theorem]{Lemma}

\theoremstyle{definition}
\newtheorem{definition}[theorem]{Definition}
\newtheorem*{definition*}{Definition}

\newtheorem*{assumption*}{Assumption}
\theoremstyle{remark}

\newtheorem*{remark*}{Remark}
\newtheorem*{notation*}{Notational remark}
\newtheorem{example}[theorem]{Example}
\newtheorem*{case*}{Case}

\DeclareMathOperator{\ex}{\debug{\mathbb{E}}}

\DeclareMathOperator{\expo}{\debug e}

\DeclareMathOperator{\prob}{\debug{\mathbb{P}}}

\DeclareMathOperator{\supp}{\debug{supp}}

\newcommand{\naturals}{\mathbb{\debug N}}
\newcommand{\reals}{\mathbb{\debug R}}

\newcommand{\simplex}{\debug \Delta}

\newcommand{\diff}{\ \textup{\debug d}}

\newcommand{\proba}{\debug p}

\newcommand{\state}{\debug \theta}
\newcommand{\statealt}{\state'}
\newcommand{\State}{\debug \Theta}
\newcommand{\states}{\mathcal{\debug T}}
\newcommand{\prior}{\debug \mu}
\newcommand{\distr}{\debug F}
\newcommand{\Bor}{\mathcal{\debug B}}
\newcommand{\dirac}{\debug \delta}

\newcommand{\per}{\debug t}
\newcommand{\randper}{\debug \tau}
\newcommand{\history}{\debug h}

\newcommand{\graph}{\mathcal{\debug G}}
\newcommand{\network}{\mathcal{\debug N}}

\newcommand{\vertex}{\debug v}
\newcommand{\vertexalt}{\debug u}
\newcommand{\vertices}{\mathcal{\debug V}}
\newcommand{\edge}{\debug e}
\newcommand{\edges}{\mathcal{\debug E}}

\newcommand{\route}{\debug r}
\newcommand{\routealt}{\route'}
\newcommand{\routes}{\mathcal{\debug R}}

\newcommand{\capac}{\debug \gamma}
\newcommand{\cut}{\mathcal{\debug C}}

\newcommand{\nodeA}{\debug a}
\newcommand{\nodeB}{\debug b}
\newcommand{\nodeC}{\debug c}

\newcommand{\flow}{\debug y}
\newcommand{\flowprof}{\boldsymbol{\flow}}
\newcommand{\flows}{\mathcal{\debug Y}}

\newcommand{\load}{\debug x}
\newcommand{\loadprof}{\boldsymbol{\load}}

\newcommand{\demand}{\debug d}
\newcommand{\Demand}{\debug D}

\newcommand{\source}{\mathsf{\debug O}}
\newcommand{\sink}{\mathsf{\debug D}}

\newcommand{\game}{\debug \Gamma}

\newcommand{\nagame}{\debug G}
\newcommand{\dgame}{\debug \Gamma}

\newcommand{\play}{\debug i}

\newcommand{\cost}{\debug c}
\newcommand{\costprof}{\boldsymbol{\cost}}

\newcommand{\eq}[1]{#1^{\ast}}

\newcommand{\incomp}[1]{\widetilde#1}
\newcommand{\compl}[1]{\widetilde#1}

\newcommand{\argdot}{\,\cdot\,}

\newcommand{\consta}{\debug a}
\newcommand{\consteps}{\debug \varepsilon}
\newcommand{\bigconst}{\debug A}
\newcommand{\good}{\mathsf{\debug G}}
\newcommand{\bad}{\mathsf{\debug B}}
\newcommand{\learnd}{\debug L}

\newcommand{\interval}{\mathcal{\debug I}}
\newcommand{\demandint}{\mathcal{\debug D}}
\newcommand{\convalpha}{\debug \alpha}

\newcommand{\noisy}[1]{\tilde#1}
\newcommand{\noedges}{\debug k}
\newcommand{\capkappa}{\debug \kappa}

\newacro{eNE}[$\varepsilon$-NE]{$\varepsilon$-Nash equilibrium}
\newacroplural{eNE}[eNE]{$\varepsilon$-Nash equilibria}
\newacro{NE}{Nash equilibrium}
\newacro{NEs}{Nash equilibria}
\newacroplural{NE}[NE]{Nash equilibria}
\newacro{PNE}{pure Nash equilibrium}
\newacroplural{PNE}[PNE]{Nash equilibria}
\newacro{PFNE}{prior-free Nash equilibrium}
\newacroplural{PFNE}[PFNE]{weakly prior-free Nash equilibria}
\newacro{WPFNE}{weakly prior-free Nash equilibrium}
\newacroplural{WPFNE}[WPFNE]{weakly prior-free Nash equilibria}
\newacro{WE}{Wardrop equilibrium}
\newacroplural{WE}[WE]{Wardrop equilibria}

\newacro{SO}{socially optimum}

\newacro{NRG}{nonatomic routing game}
\newacro{RGUNS}{nonatomic routing game with unknown network state}
\newacro{DRGUNS}{dynamic nonatomic routing game with unknown network state}
\newacro{SP}{series-parallel}

\newacro{iid}[i.i.d.]{independent and identically distributed}

\title{Social Learning in Nonatomic Routing Games}

\author{Emilien Macault$^{\ddag}$}
\address{$^{\ddag}$ HEC Paris, 1 Rue de la Lib\'eration, 78350 Jouy-en-Josas, France}
\EMAIL{emilien.macault@hec.edu}

\author{Marco Scarsini$^{\P}$}
\address{$^{\P}$ Dipartimento di Economia e Finanza, Luiss University, Viale Romania 32, 00197 Roma, Italy.}
\EMAIL{marco.scarsini@luiss.it}

\author{Tristan Tomala$^{\S}$}
\address{$^{\S}$ HEC Paris and GREGHEC, 1 rue de la Lib\'eration, 78351 Jouy-en-Josas, France; 
}
\EMAIL{tomala@hec.fr}

\thanks{Declarations of interest: none}

\begin{document}

\begin{abstract}

We consider a discrete-time nonatomic routing game with variable demand and uncertain costs. 
Given a routing network with single origin and destination, the cost function of each edge depends on some uncertain persistent state parameter. 
At every period, a random traffic demand is routed through the network according to a \acl{WE}. 
The realized costs are publicly observed and the public Bayesian belief about the state parameter is  updated. 
We say that there is \emph{strong learning} when beliefs converge to the truth and \emph{weak learning} when the equilibrium flow converges to the complete-information flow. 
We characterize the networks for which learning occurs. 
We prove that these networks have a series-parallel structure and provide a counterexample to show that learning may fail in non-series-parallel networks.


\end{abstract}
\maketitle
\pagebreak

\section{Introduction}

Nowadays navigation systems  provide real-time global information about congestion and the state of roads in the network to every driver, using users' data. Yet, even if information is publicly and completely broadcast (all drivers have perfect information about traffic), the amount of gathered information  might not be  socially efficient. This opens up the following question: `in a dynamic routing game with incomplete information, is it possible that social learning emerges from short-lived agents' equilibrium behavior?'
Routing games model the behavior of selfish agents who choose one path on a network from their origin to their destination, with the goal to minimize their cost, identified with the traveling time. 
This traveling time depends on the choice of all players, since the cost of an edge increases with the number of agents who use it. 
Routing games with many agents are often approximated by nonatomic games, which are more tractable and represent the limit case where each agent is negligible.
In most of the existing literature, costs are assumed to be known, but in reality they are affected by unpredictable circumstances. 
Uncertainty relative to these cost functions may strongly impact the equilibrium behavior, attracting agents away from potentially optimal actions. Thus the analysis of routing games of incomplete information is an important object of study.

If a game where  cost functions are not known is repeated over time and  beliefs of  players are updated taking into account observations of  previous players, will the costs functions be eventually learned? 
Two opposite effects naturally arise: on  one hand, agents aim at  minimizing the costs they incur immediately. On the other hand, socially efficient behavior requires agents to explore the routing network in order to learn the actual costs. 
We  consider generations of short-lived players who play the game only once and  are being replaced every period by a new set of players. 
Some amount of social learning may be achieved as players of one generation update their beliefs based on the behavior of the previous generations. 
Thus, selfish behavior may provide public information for the next generations of players. 
One challenge  is to analyze the amount of such public information provision. 
If the game parameters   are stationary, given that each generation has no incentive to be forward looking, potentially informative behavior may be off  equilibrium path and thus social learning may fail. 
Yet, when  there is  variability in the circumstances in which the game is played, current equilibrium behavior may provide useful information to the subsequent generations of players.

\subsection{Our contribution}
\label{suse:our}
We  consider a  repeated symmetric \acl{NRG} where each edge has a capacity and the cost functions of each edge depend on the load of the edge and on an unknown state parameter that is invariant over time.  
The set of  states  is finite and  endowed with a common prior.
At each period of time, a short-lived generation of users  with a given total demand plays the game and realizes a \acl{WE} with respect to the expected costs on edges: each path that receives positive load has the least expected cost.
For every used edge, its load and the corresponding realized cost become public information for the following generations. 
There is perfect recall, so each generation knows the entire past history of the game and updates its beliefs in a Bayesian way.
The sequence of different generations' demands is assumed to be random, \ac{iid}.

We consider two concepts of social learning:
under strong learning, players eventually learn the true state of the world;
under weak learning they learn to play the game as if the true state of the world were known. 
We show that weak learning is a strictly weaker concept than strong learning and that the conditions to achieve either of them depend on the topology of the network and on the support of the random demand. 
Our main theorem proves that weak learning occurs if the routing network is series-parallel, the cost functions are unbounded, and the demand can potentialy reach the network capacity. 
Further, we show that strong learning is achieved under the same prerequisites and the additional condition that the demand has full support. 
The intuition behind this result is the following: when the demand is stochastic, equilibrium flows vary. This generates observations of the cost functions for different values of loads. 
Based on results from \citet{ComDosSca:MPB2021} on the variation of equilibrium flows with respect to the demand, we prove that in a series-parallel network, as the demand tends to congest the network, all edges are used in equilibrium and equilibrium loads can get arbitrarily close to capacity. 
This implies that, with probability one, the cost functions will be observed at levels which allow distinguishing between the cost-relevant states.
An important subtlety in the learning result is that traveling an edge does not automatic reveal its cost function, since two different cost functions may be equal on an entire load interval. 
To distinguish the cost functions we need an equilibrium load where they are not the same.
Finally, we prove that the condition on the network topology is necessary: for  networks that do not satisfy it, it is possible to construct a game where even weak learning fails. 
The intuition is that a network which is not series-parallel contains a Wheatstone sub-network, and  the topology of Wheatstone network is such that some edges are not used in equilibrium when the demand is high.

\subsection{Literature review}

\citet{Ban:QJE1992,BicHirWel:JPE1992} considered models where agents enter a market sequentially, update their beliefs by taking into consideration their private signals and the actions chosen by the previous agents, and make their optimal decisions accordingly. 
They show that social learning may fail, that is, it is possible that in equilibrium all agents choose a suboptimal action.
\citet{SmiSor:E2000} showed that this is due to the hypothesis that the private signals are bounded, so, from some point on, no private signal can overcome the the observations' strength.
When signals are unbounded, social learning occurs.
A very general version of this model was recently studied by \citet{AriMue:MORfrth}.

In this paper we aim at lifting the concept of social learning to strategic models of traffic congestion. 
In the nonatomic setting that we adopt, the role of agents is taken by a sequence of traffic demands.
At every period a \acl{NRG} is played with a different demand generation.
Although an example of traffic congestion game can be found in \citet{Pig:Macmillan1920}, the standard solution concept for \aclp{NRG} is due to \citet{War:PICE1952}.
Some interesting properties of this concept were then studied by  \citet{BecMcGWin:Yale1956}.
Every \acl{NRG} is a congestion games, which in turn is potential games. 
Atomic congestion games were introduced by \citet{Ros:IJGT1973}, who proved that any congestion game admits a potential function, whose local minimizers are the pure Nash equilibria of the game.
\citet{MonSha:GEB1996} studied potential games and some of their generalizations and proved that for any potential game there exists a congestion game with the same potential function.
Congestion games with a continuum of players and their relation to potential games were studied by \citet{San:JET2001}.
The details of the asymptotic relation between atomic and nonatomic congestion games have been examined by \citet{ComScaSchSti:Converg:arXiv2019}.
\citet{ScaTom:IJGT2012} studied repeated versions of routing games.

Informational issues in routing games have been considered by several authors.
Some of them deal with atomic games. 
For instance, \citet{GaiMonTie:TCS2008} considered atomic routing games with incomplete information  where the type of a player is her traffic, which is private information. 
They proved existence of Bayesian pure equilibria and they studied their complexity.
They then studied properties of completely mixed Nash equilibria and finally they provided bounds for the price of anarchy.
\citet{Gai:AOA2009} studied atomic congestion games where each player can be of two types, rational or malicious. 
He proved that pure Bayesian equilibria may fail to exist and studied the price of malice.
\citet{AshMonTen:AI2009} studied symmetric congestion games where the number of active players is unknown and there is no known prior distribution over the number of active players.
\citet{BerSch:TCS2010} considered evolutionary stable strategies for Bayesian routing games with parallel links.
\citet{FotGkaKapSpi:TCS2012} studied how the inefficiency of equilibria in congestion games is affected by what they call social ignorance, that is, the lack of information about the presence of other players.
\citet{Syr:arXiv2012} and \citet{Rou:ACMTEC2015} proved that known bounds for the price of anarchy in smooth games extend to their incomplete information version when players' private preferences are drawn independently.
\citet{Rou:ACMTEC2015} showed that the extension does not hold for correlated preferences.
\citet{ComScaSchSti:StocPoA:arXiv2019} studied the behavior of the price of anarchy in atomic congestion games where each player $\play$ takes part in the game with some probability $\proba_{\play}$ and the participations are independent. 
\citet{GaiTar:EC2020,GaiTar:EC2021} examined discrete-time queueing models where routers compete for servers and learn using strategies that satisfy a no-regret condition. 
The key element of their model is the explicit consideration of carryover effect from one period to the other.

Other papers deal with nonatomic routing games.
Some of them deal with uncertainty in the traffic demand.
\citet{WanVinChe:TRB2014} examined nonatomic routing games with random demand and studied the dependence of the price of anarchy on the variability of the demand and on the network structure. 
\citet{CorHoeSch:TRB2019} studied a class of nonatomic routing games where different sets of players take part in the game with some probability; they considered the Bayes Nash equilibrium of these games and showed that the bounds for price of anarchy do not deteriorate with respect to the complete information version of the game.
\citet{BhaLigSchSwa:GEB2019} studied non-atomic routing games where cost functions are unknown. 
Their goal was to determine edge tolls in order to achieve specific equilibria. 
They showed that this can be done under mild conditions through the use of an oracle computing equilibrium flows given a set of tolls. 
In particular, they computed tight complexity bounds for series-parallel routing networks.

\citet{AceMakMalOzd:OR2018} dealt with nonatomic routing games where different types of agents have different information sets and each agent can only use paths in her own information set. 
They considered non-oriented routing networks and defined the concepts of information-constrained Wardrop equilibrium  and of informational Braess’ paradox, that is, a situation where, if agents get more information, they experience a higher cost in equilibrium. 
They showed that a necessary and sufficient condition for the paradox not to happen is to have   a network in the series of linearly independent class.
\citet{WuLiuAmi:ACC2017} studied a routing game where agents subscribe to one of two traffic information systems and characterized equilibria under two information structures whose difference is the assumption of the common prior in one and not in the other.
In the routing game studied by \citet{WuAmiOzd:OR2021} the population of users is divided into groups and each group subscribes to specific traffic information system. 
The cost functions on each edge are state dependent and each traffic information system sends a noisy signal only to its subscribers.
The solution concept that they adopted is the Bayesian Wardrop equilibrium.
They studied the sensitivity of the equilibria with respect to changes in the size of the groups.
The paper by \citet{WuAmi:IFAC2019} is the closest in spirit to our own. 
They analyzed nonatomic routing games with unknown costs where Bayesian public beliefs are  updated over time and only the used edges provide some information about the realized costs.
The difference between their model and ours is that they consider  constant demand and noisy costs with  Gaussian noise, possibly correlated across different edges. 
While close in flavor, the two different sets of assumptions produce different learning outcomes. 
In our model, costs are deterministic functions of a random state of the world which is fixed \emph{ex-ante}. This implies that learning revolves around sampling the cost functions at sufficiently many load levels in order to accurately distinguish the different possible states. This allow us to  find  conditions on distribution of demand and network topology that ensure learning. The differences between the two models are detailed in \cref{se:signals} where we provide instances of routing games where noisy costs cannot induce learning whereas a random demand does.

There is an important body of literature on various forms of opinion merging, partial learning, and the interaction between equilibrium and beliefs dynamics. 
\citet{KalLeh:JME1994} defined and characterized weak and strong merging of opinions. 
Closer to the problem examined here, \citet{KalLeh:E1993} defined rational learning for an infinitely repeated game where Bayesian players with heterogeneous beliefs maximize their expected utilities. 
They showed that, if agents' prior belief and the truth satisfy an ``absolute continuity'' condition, then the sequence of plays converges to an outcome arbitrarily close to a Nash equilibrium. 
Yet, in  general, even if agents perfectly observe  action profiles, errors in prior beliefs may persist over the course of the game.
\citet{FudLev:E1993} introduced the concept of self-confirming equilibrium, where players' beliefs about other players' actions are required to be correct only on the equilibrium path. As agents play best responses to their beliefs,  these may differ substantially from the truth, as long as nothing contradicts the beliefs. Hence, self-confirming equilibria and Nash equilibria of a game need not be the same. 
Given the sequential nature of the repeated game, self-confirming equilibria may allow players to hold beliefs on other agents that may be inconsistent with rational behavior.
\citet{BatGua:DGM1997} refined self-confirming equilibrium by proposing the concept of conjectural equilibrium for extensive form games of incomplete information without a common prior. At a conjectural equilibrium, players are assumed to behave rationally given the information they have about the game parameters. 
\citet{RubWol:GEB1994} pushed this idea further through the concept of rationalizable conjectural equilibrium, which requires that agents' rationality be common knowledge. 
In this paper, we adopt a similar point of view. 
Namely, we study the steady states of social learning dynamics where players myopically best-respond to the information obtained by previous generations. 
We show that this information is correct  for the edges of the network that have been explored along the equilibrium path,  but may remain incorrect for other edges.  
Thus, some paths that would be used under full information may remain unused forever.

\subsection{Outline of the paper}
\label{suse:outline}
\cref{se:routing-games-incomplete-info} introduces \aclp{RGUNS}. 
\cref{se:dynamic-games} deals with \aclp{DRGUNS}. 
\cref{se:learning} studies conditions for learning to occur. 
All the proofs can be found in \cref{se:proofs}. 
\cref{se:signals} provides a comparison with other instances of learning in routing games.

\section{Routing games with unknown network state}
\label{se:routing-games-incomplete-info}

We first describe the baseline model of nonatomic routing game. 
A \emph{network} is an oriented multigraph $\network=(\vertices,\edges)$, where $\vertices$ is the vertex set and  $\edges$ is the edge set, endowed with an origin/destination pair  $\source,\sink\in\vertices$.
Given $\vertex,\vertexalt\in\vertices$, a  \emph{path} from $\vertex$ to $\vertexalt$ is an ordered set of edges $\edge_{1},\dots\edge_{k}$ such that the tail of $\edge_{1}$ is $\vertex$, the head of $\edge_{k}$ is $\vertexalt$, and, for each $i\in\braces{1,\dots,k-1}$, the head of $\edge_{i}$ is the tail of $\edge_{i+1}$.
The set $\routes$ indicates the set of paths from $\source$ to $\sink$. 
To avoid trivialities, we assume that each edge is part of a path in $\routes$.
Each $\edge\in\edges$ is endowed with a capacity $\capac_{\edge} \in(0,+\infty]$ and with a  continuous strictly increasing function $\cost_{\edge} \colon [0,\capac_{\edge})\to\reals_{+}$   that represents the cost of using edge $\edge$ as a function of its load.
The traffic demand is denoted by $\demand$. 
The tuple $\nagame\coloneqq\parens*{\demand,\network,\braces{\capac_{\edge}}_{\edge\in\edges},\braces{\cost_{\edge}}_{\edge\in\edges}}$ defines a \emph{\acl{NRG}}.

A  \emph{cut} $\cut$ of the network is a subset of $\edges$ such that there does not exist a path from $\source$ to $\sink$ the uses only edges in $\edges\setminus\cut$.
The capacity of $\cut$ is the sum of the capacities of its edges $\capac_{\cut} \coloneqq \sum_{\edge\in\cut}\capac_{\edge}$. 
The \emph{capacity} $\capac$ of the network $\network$ is the   smallest capacity among all possible cuts; it  corresponds to the maximum traffic that can flow from origin to destination
\citep[see][]{ForFul:PUP1962}.
Throughout the paper, we assume $\demand \in [0,\capac)$, that is the demand satisfies the capacity constraints.
If $\capac=\infty$, then any positive demand can be satisfied.

For each path $\route\in\routes$, $\flow_{\route}\in \reals_{+}$ denotes the \emph{flow} over path $\route$.
For each edge $\edge\in\edges$, the \emph{load} $\load_{\edge}$ of $\edge$ is defined as
\begin{equation}
\label{eq:load}
\load_{\edge} \coloneqq \sum_{\route\ni\edge}\flow_{\route}.
\end{equation}
The symbols $\loadprof=\braces{\load_{\edge}}_{\edge\in\edges}$ and $\flowprof=\braces{\flow_{\route}}_{\route\in\routes}$ denote the load vector and the flow vector, respectively.
A flow vector $\flowprof$ is \emph{feasible} if it satisfies the demand and obeys the capacity constraints, i.e., 
\begin{align}
\label{eq:demand}
\sum_{\route\in\routes}\flow_{\route}&=\demand,\\
\load_{\edge} &< \capac_{\edge}, \quad\text{for all }\edge\in\edges.
\end{align}
The set of \emph{feasible flows} is denoted by $\flows$.
Notice that $\loadprof$ is uniquely determined by $\flowprof$, but not vice versa.

The cost of using edge $\edge$ is $\cost_{\edge}(\load_{\edge})$, with an abuse of notation  the cost of using path $\route$ is denoted by
\begin{equation}
\label{eq:cost-path}
\cost_{\route}(\flowprof) \coloneqq \sum_{\edge\in\route}\cost_{\edge}(\load_{\edge}).
\end{equation}
The cost vector $\parens{\cost_{\edge}(\load_{\edge})}_{\edge\in\edges}$ induced by the load vector $\loadprof$ is denoted by $\costprof(\loadprof)$.

This model encompasses the classical case where there is no limit of capacity: $\capac_\edge=+\infty$ for each $\edge$. 
It also covers M/M/1 queuing models where the cost (or waiting time) tends to infinity as the demand approaches the capacity. 
An alternative approach would be to extend the cost functions to the whole $\reals_{+}$ and allow values in $[0,+\infty]$; this would have no impact on our results.

\begin{definition}
\label{de:Wardrop-equilibrium}
A flow $\eq{\flowprof}\in\flows$ is a \acfi{WE}\acused{WE} of $\nagame$ if, for all $\route,\routealt\in\routes$ with $\eq{\flow}_{\route}>0$, we have
\begin{equation}
\label{eq:Wardrop}
\cost_{\route}(\eq{\flowprof}) \le \cost_{\routealt}(\eq{\flowprof}).
\end{equation}
\end{definition}

Since cost functions are strictly increasing, there exists a unique  load $\eq{\loadprof}$ associated to any equilibrium flow $\eq{\flowprof}$ (this unique load is the minimizer of a strictly convex potential function).

We now introduce uncertainty, represented by a finite probability space $\parens*{\states,2^{\states},\prior}$ called the \emph{state space}.
A \emph{\acl{RGUNS}} 
is a tuple $\nagame_{\prior}\coloneqq\parens*{\nagame,\parens*{\states,2^{\states},\prior}}$, where $\nagame$ is a \acl{NRG} as before and, for each $\edge\in\edges$ and  each $\state\in\states$, the cost function $\load\mapsto\cost_{\edge}(\load,\state)$  is continuous and strictly   increasing on $\load\in [0,\capac_{\edge})$.
To guarantee identifiability of the states, we  assume that for every pair $\state,\statealt\in\states$, there exists an edge $\edge\in\edges$ such that $\cost_{\edge}(\argdot,\state)\neq\cost_{\edge}(\argdot,\statealt)$.

Given a prior distribution $\prior\in\simplex(\states)$, with a common abuse of notation,  the expected costs are denoted by
\begin{equation}
\label{eq:cost-prior}
\cost_{\edge}(\load,\prior) \coloneqq \int_{\states} \cost_{\edge}(\load,\state) \diff\prior(\state)
\quad\text{and}\quad
\cost_{\route}(\flowprof,\prior) \coloneqq \int_{\states}\cost_{\route}(\flowprof,\state)\diff\prior(\state).
\end{equation}

\begin{definition}
\label{de:Bayes-Wardrop-equilibrium}
A flow $\eq{\flowprof}\in\flows$ is a \acfi{WE}\acused{WE} of $\nagame_{\prior}$ if, for all $\route,\routealt\in\routes$ with $\eq{\flow}_{\route}>0$, we have
\begin{equation}
\label{eq:Bayes-Wardrop}
\cost_{\route}(\eq{\flowprof},\prior) \le \cost_{\routealt}(\eq{\flowprof},\prior).
\end{equation}
In \cref{de:Bayes-Wardrop-equilibrium}  the Wardrop equilibrium of  the game is defined in terms of  the expected cost functions. Note that strict monotonicity of the  cost functions implies uniqueness of the equilibrium load.
The \acp{WE} load of the game $\nagame_{\prior}$ with demand $\demand$ is denoted by $\eq{\loadprof}(\demand,\prior)$.
\end{definition}

\section{Dynamic routing games with unknown state}
\label{se:dynamic-games}

We now consider a discrete-time model of social learning where a routing game with unknown network state is played over time and with a population that changes at every period.
The goal is to find conditions under which social learning is achieved, that is, each generation learns from the behavior of the previous generations and the public beliefs about the state of nature converge to the true value.

If the traffic demand were constant over time, then the same \ac{WE} would end up being played every period.
As a consequence, learning does not occur whenever some equilibrium paths of the complete-information games are not used along the sequence of equilibria.
Therefore, to achieve learning, we will assume that  demands are given by a sequence $\parens{\Demand^{\per}}_{\per\in\naturals}$ of \ac{iid} nonnegative random variables with common marginal distribution, denoted by $\distr$.
The symbol $\Demand$ denotes a generic element of the sequence and $\supp(\Demand)$ denotes its support. We  assume independence between the demands and the state of nature.
Therefore $\distr$ and $\prior$ induce a unique product measure $\prob$ on the measurable product space $\parens*{[0,\capac)^{\infty}\times\states,\Bor([0,\capac)^{\infty})\otimes 2^{\states}}$.

At every period $\per$, a demand $\Demand^{\per}$ is realized and observed, the \ac{WE} is played, the equilibrium load profile ${\eq{\loadprof}}^{\per}$ and the equilibrium costs $\costprof(\eq{\loadprof},\state)=\parens{\cost_{\edge}({\eq{\load}}^{\per},\state)}_{\edge\in\edges}$ are observed.
Therefore, for every $\per=1,2,\dots$, the history at period $\per$ is
\begin{equation}
\label{eq:history-t}
\history^{\per} \coloneqq \parens*{\Demand^{1},{\eq{\loadprof}}^{1}(\Demand^{1}),\costprof\parens*{{\eq{\loadprof}}^{1}(\Demand^{1}),\State},\dots,\Demand^{\per-1},{\eq{\loadprof}}^{\per-1}(\Demand^{\per-1}),\costprof\parens*{{\eq{\loadprof}}^{\per-1}(\Demand^{\per-1}),\State},\Demand^{\per}}
=\parens*{\incomp{\history}^{\per-1},\Demand^{\per}},
\end{equation}
where $\State$ is the random state.

The distribution $\prior$ on the state space is updated according to Bayes rule and $\prior^{\per}$ denotes the posterior distribution $\prior(\argdot\mid\history^{\per})$, that is,
\begin{equation}
\label{eq:posterior}
\prior^{\per}(\state) \coloneqq \prob\parens*{\State=\state\mid\history^{\per}}.
\end{equation}
The pair $\parens{\nagame,\prob}$ defines a \emph{\acl{DRGUNS}}, which will be denoted by $\dgame$.

The sequence of posterior beliefs is a bounded martingale. 
Thus, by the martingale convergence theorem, there exists a random variable  $\prior^{\infty}$ such that $\prior^{\per}\to\prior^{\infty}$ almost surely. 
Moreover, since the set of states is finite and costs are deterministic functions of states and loads, there exists a random time $\randper$ such that almost surely, $\prior^{\per}=\prior^{\infty},\forall\per\ge\randper$. 
Indeed in every period, either $\prior^{\per}=\prior^{\per+1}$ or the support of $\prior^{\per+1}$ is a strict subset of the support of $\prior^{\per}$.

As mentioned before, we will find conditions under which social learning is achieved.
We now define two concepts of learning.
To this end, the Dirac measure on $\state$ will be denoted by $\dirac_{\state}$.

\begin{definition}
\label{de:learning}
Consider a \acl{DRGUNS} $\dgame$. We say that:
\begin{enumerate}[(a)]
\item
\label{it:strong-learning}
\emph{strong learning} is achieved if 
\begin{equation}
\label{eq:strong-learning}
\prior^{\infty}=\dirac_{\State} \quad\prob\text{-a.s.}.
\end{equation}

\item
\label{it:weak-learning}
\emph{weak learning} is achieved if 
\begin{equation}
\label{eq:weak-learning}
\eq{\loadprof}(\argdot,\prior^{\infty}) = \eq{\loadprof}(\argdot,\dirac_{\State})
\quad\prob\text{-a.s.}.
\end{equation}
\end{enumerate}
\end{definition}

The idea of \cref{de:learning} is the following. 
If strong learning is achieved, asymptotically the true state of the world is discovered. 
Actually, given the fact that the state space $\states$ is finite, strong learning implies that there exists a random time $\randper$ such that $\prior^{\per}=\dirac_{\State}$ for all $\per \ge \randper$.
Under weak learning the true state is not necessarily discovered, but, asymptotically, in equilibrium, traffic is routed as if the state were known. 
This distinction follows \citet{KalLeh:E1993}, who studied the convergence of beliefs in repeated games and showed that players may learn to predict their opponents' actions even though they do not learn their payoff matrices. 
Similarly, in our model, players may not identify the true state of the world, yet they may learn to play optimally conditional on this state.
It is not difficult to see that strong learning implies weak learning, but the converse implication is false. 
We will show this with the following example.

\begin{figure}[h]
\centering
\begin{tikzpicture}[thick]
	\tikzset{SourceNode/.style = {draw, circle, fill=green!20}}
	\tikzset{MidNode/.style = {draw, circle}}
	\tikzset{DestNode/.style = {draw, circle, fill=green!20}}
		\node[SourceNode](s) at (-4,0) {$\source$};
    \node[MidNode](a) at (0,1.5) {$\nodeA$};
		\node[MidNode](b) at (0,-1.5){$\nodeB$};
		\node[DestNode](c) at (4,0){$\sink$};
    \draw[->](s)to node[fill = white]{$\edge_{1}$}(a);
		\draw[->](s)to node[fill = white]{$\edge_{2}$}(b);
		\draw[->](a)to node[fill = white]{$\edge_{4}$}(c);
		\draw[->](b)to node[fill = white]{$\edge_{5}$}(c);
		\draw[->](a)to node[fill = white]{$\edge_{3}$}(b);
\end{tikzpicture}
\caption{\label{fi:Wheatstone} Wheatstone network.}
\end{figure}
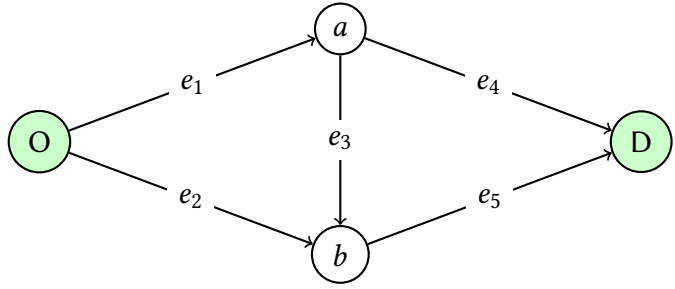

\begin{example}
\label{ex:weak-vs-strong}
Consider the network in \cref{fi:Wheatstone}  
 where each edge has infinite capacity and 
\begin{align}
\cost_{1}(\load,\state)=\cost_{5}(\load,\state)=\load,\quad
\cost_{2}(\load,\state)=\cost_{4}(\load,\state)=1+\consteps\load,\quad
\cost_{3}(\load,\state)=
\begin{cases}
\consteps\load & \text{for }\state=\state_{\good},\\
10+\consteps\load & \text{for }\state=\state_{\bad}
\end{cases}
\end{align}
with $\prior(\state_{\bad})=\prior(\state_{\good})=1/2$ and $\consteps<1$.  
Let the demand $\Demand^{\per}$ have a  distribution with  support $[20,\infty)$.

In this setup, for every value of the demand in the support, edge $\edge_{3}$ is not used. 
To see this, consider the paths
\begin{equation}
\label{eq:paths}
\route_{1}=(\edge_{1},\edge_{4}),\quad
\route_{2}=(\edge_{2},\edge_{5}),\quad
\route_{3}=(\edge_{1},\edge_{3},\edge_{5}),
\end{equation}
and let $\flow_{1}, \flow_{2}$ and $\flow_{3}$ be their respective equilibrium flows under $\prior$. 
Observe that for any value of $\Demand^{\per}$,  $\flow_{1}=\flow_{2}$ by symmetry. 
Then, given a realization $\demand$ of $\Demand^{1}$, there is a positive flow on $\route_3$ if and only if $\cost_{\route_3}(0,\prior)
\le \cost_{\route_1}(\demand/2,\state)$, that is, 
\begin{equation*}
\demand+5\le \frac{\demand}{2}(1+\consteps)+1, \quad\text{i.e.,}\quad \frac{\demand/2+4}{\demand/2} \le \consteps,
\end{equation*}
which is impossible. 
Thus, no demand is routed on $\route_3$, edge $\edge_3$ remains unexplored, and the state is not identified at time $1$.
At time $2$ the situations is the same, so, by induction, $\prior^{\per}=\prior$ for every $\per\in\naturals$.
On the other hand, observe that even under full information, $\edge_{3}$ would not be used given that $\Demand^{\per}\ge 20$. 
This shows that weak learning is trivially achieved, but strong learning is not.
\end{example}

\section{Learning in the Dynamic Game}
\label{se:learning}

Our convergence result  requires assumptions on the cost functions, on the sequence of demands, and on the structure of the network. 
As is shown later, these assumptions are necessary.

\begin{definition}
\label{de:series-parallel}
A network $\network$ is called \acfi{SP}\acused{SP} if it can be defined sequentially as follows:
\begin{enumerate}[(a)]
\item
Either $\network$ has a single edge.

\item
Or $\network$ consists of two \ac{SP} networks connected in series, by merging the destination of the first with the origin of the second.

\item
Or $\network$ consists of two \ac{SP} networks connected in parallel, by merging the origin of the first with the origin of the second and the destination of the first with the destination of the second.
\end{enumerate}
\end{definition}

We refer the reader to \citet{RioSha:JMPMIT1942,Duf:JMAA1965,HolLaw:MSS2003,Mil:GEB2006} for the definition and properties of \ac{SP} networks. 
\cref{fi:sp} provides examples of \ac{SP} and non-\ac{SP} networks.

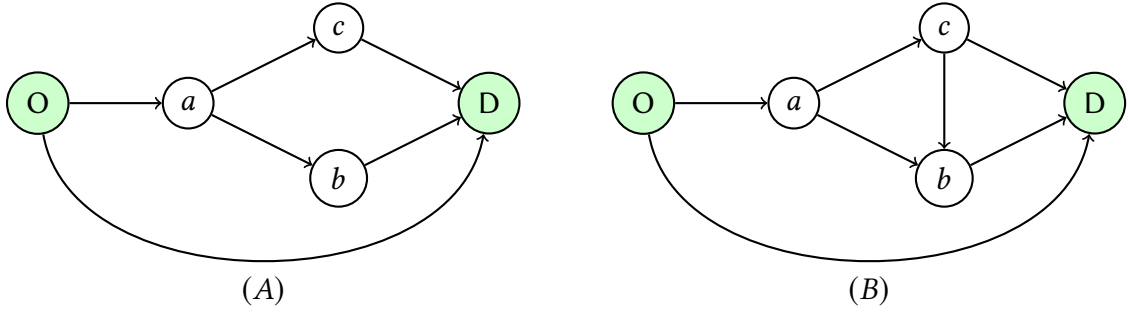
\begin{figure}
\begin{center}
\begin{tikzpicture}[thick]
	\tikzset{SourceNode/.style = {draw, circle, fill=green!20}}
	\tikzset{MidNode/.style = {draw, circle}}
	\tikzset{DestNode/.style = {draw, circle, fill=green!20}}
		\node[SourceNode](s) at (-4,0) {$\source$};
    \node[MidNode](a) at (-2,0) {$\nodeA$};
		\node[MidNode](c) at (0,-1){$\nodeB$};
		\node[MidNode](b) at (0,1){$\nodeC$};
		\node[DestNode](t) at (2,0){$\sink$};
    \draw[->](s)--(a);
		\draw[->](a)--(b);
		\draw[->](a)--(c);
		\draw[->](b)--(t);
		\draw[->, bend right = 80](s) to node{} (t);
		\draw[->](c)--(t);
		\node at (-1,-2.5) {$(A)$};
\end{tikzpicture}
\hspace{1cm}
\begin{tikzpicture}[thick]
	\tikzset{SourceNode/.style = {draw, circle, fill=green!20}}
	\tikzset{MidNode/.style = {draw, circle}}
	\tikzset{DestNode/.style = {draw, circle, fill=green!20}}
		\node[SourceNode](s) at (-4,0) {$\source$};
    \node[MidNode](a) at (-2,0) {$\nodeA$};
		\node[MidNode](c) at (0,-1){$\nodeB$};
		\node[MidNode](b) at (0,1){$\nodeC$};
		\node[DestNode](t) at (2,0){$\sink$};
    \draw[->](s)--(a);
		\draw[->](a)--(b);
		\draw[->](a)--(c);
		\draw[->](b)--(c);
		\draw[->](b)--(t);
		\draw[->, bend right = 80](s) to node{} (t);
		\draw[->](c)--(t);
		\node at (-1,-2.5) {$(B)$};
\end{tikzpicture}
\end{center}
\caption{\label{fi:sp} Network $(A)$ is \ac{SP}. Network $(B)$ is not \ac{SP}, due to the edge from node $\nodeC$ to node $\nodeB$.}
\end{figure}

The next  theorem provides conditions for weak and strong learning.

\begin{theorem}
\label{th:learning}
Let  $\dgame$ be a \acl{DRGUNS} such that, the network $\network$ is \ac{SP} and, for each edge $\edge\in\edges$ and each $\state\in\states$, we have $\lim_{\load \to \capac_{\edge}}\cost_{\edge}(\load,\state)=+\infty$.

\begin{enumerate}[\upshape(a)]

\item
\label{it:th:learning-strong}
If  $\supp(\Demand)=[0,\capac)$, then strong learning occurs.%
\footnote{By definition, the support of a random variable is a closed set. Here it is closed  relative to the space of feasible demands $[0,\capac)$.}

\item
\label{it:th:learning-weak}
If for every $\consteps>0$, there exists $\demand \in (\capac-\consteps,\capac)$ such that $\demand \in \supp(\Demand)$, then weak learning occurs.

\end{enumerate}
\end{theorem}

The following theorem shows that the assumption that the network is \ac{SP} is necessary in the sense that, if it does not hold, it is possible to construct an assignment of cost functions and capacities that satisfies the other hypotheses of \cref{th:learning} and for which weak learning fails for any distribution of the demand.

\begin{theorem}
\label{th:converse}

If the network $\network$ is not \ac{SP}, then there exist capacities and uncertain unbounded cost functions such that weak learning fails for every distribution of the demand.

\end{theorem}

\subsection{Learning failure}
\label{suse:learning-failure}

Learning may fail if any of the assumptions of \cref{th:learning} does not hold. 
We present below a series of examples that show why learning may fail in those cases.

\begin{figure}[h]
\centering
\begin{tikzpicture}[node distance   = 6 cm,thick]
  \tikzset{SourceNode/.style = {draw, circle, fill=green!20}}
  \tikzset{DestNode/.style = {draw, circle, fill=green!20}}
     \node[SourceNode](s){$\source$};
     \node[DestNode,right=of s](t){$\sink$};
     \draw[->,bend left = 45](s) to node[fill = white]{$\cost_{1}(\load,\state)$} (t);
     \draw[->,bend right = 45](s) to node[fill = white]{$\cost_{2}(\load,\state)$} (t);
\end{tikzpicture}
\caption{\label{fi:parallel-network} Parallel edge network.}
\end{figure}
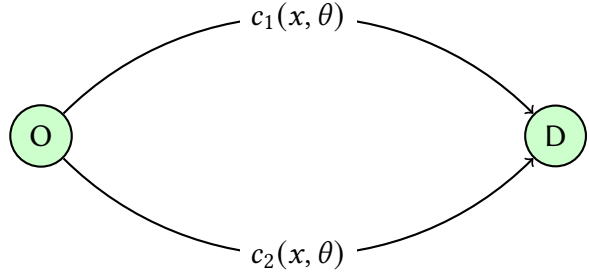

\begin{example}[Bounded costs]
\label{ex:bounded-costs}
Consider the network in \cref{fi:parallel-network} with infinite capacity on each edge and 
\begin{equation}
\label{eq:costs-Pigou-2}
\cost_{1}(\load,\state)=1-\expo^{-\load},\qquad
\cost_{2}(\load,\state)=
\begin{cases}
\load &\text{for }\state=\state_{\good},\\
\load+10 &\text{for }\state=\state_{\bad},
\end{cases}
\end{equation}
with $\prior(\state_{\bad})=\prior(\state_{\good})=1/2$.

In the unique equilibrium of the game all the demand uses edge $\edge_{1}$ at any period $\per$, for any possible value of $\Demand^{\per}$. 
This is due to the fact that the cost $\cost_1(\cdot,\state)$ is bounded above by $1$ and $\cost_2(\load,\prior)=\load+5$.
The lower path is dominated in expectation for every possible value of $\Demand^{\per}$; hence, no positive mass ever uses it at any equilibrium and the public belief remains equal to the prior.
As a consequence, weak learning does not occur. 
\end{example}

\begin{example}[Bounded demand]
\label{ex:bounded-demand}

Consider the network in \cref{fi:parallel-network} with infinite capacity on each edge and 
\begin{equation}
\label{eq:costs-Pigou}
\cost_{1}(\load,\state)=\load,\qquad
\cost_{2}(\load,\state)=
\begin{cases}
\load &\text{for }\state=\state_{\good},\\
\load + \consta &\text{for }\state=\state_{\bad},
\end{cases}
\end{equation}
with $\consta>0$ and  $\prior(\state_{\bad})=\prior(\state_{\good})=1/2$.

The expected cost of edge $\edge_{2}$ is 
\begin{equation}
\label{eq:expected-cost-e2}
\cost_{2}(\load,\prior) = \load + \frac{1}{2}\consta.
\end{equation}
Therefore, if $\Demand^{\per} < \consta/2$ with probability $1$, then edge $\edge_{2}$ is never used and even weak learning fails. 
In this example, learning fails because the lower path is dominated in expectation for low values of $\Demand^{\per}$. 
Under complete information, edge $\edge_{2}$ is used in equilibrium when the state is $\state_{\good}$.
When states are unknown, the presence of a fixed cost $\consta$ in state $\state_{\bad}$ deters the use of edge $\edge_{2}$ for low values of the demand.
Hence in equilibrium  the public belief remains equal to the prior.
\end{example}

\begin{example}[Non-SP network]
\label{ex:no-SP}
Consider the costs and network of \cref{ex:weak-vs-strong} with infinite capacity on each edge. 
The expected cost of edge $\edge_{3}$ is 
\begin{equation}
\label{eq:expected-cost-e3}
\cost_{3}(\load,\prior) = \consteps\load + 5.
\end{equation}
If the demand $\Demand^{\per}$ has an exponential distribution with parameter $1$, then edge $\edge_{3}$ is never used, no matter  the realization of $\Demand^{\per}$. Yet, it would be used for small values of the demand under $\state_{\good}$. This shows that  weak learning fails.

Since the network of this example is not series-parallel, we were able to find costs such that  edge $\edge_{3}$ is used only under complete knowledge of   state  $\state_{\good}$ and for  low values of the demand. Due to  the fixed cost in state $\state_{\bad}$, no demand will use this edge,  hence no learning occurs.
\end{example}

\section{Proofs}\label{se:proofs}

In this section we provide the proofs of our main results.

\begin{proof}[Proof of \cref{th:learning}]
Given a prior $\prior^{\per}$, we define $\learnd(\prior^{\per})$ the set of possible demands $\demand$ such that, under the equilibrium load profile $\loadprof^{*\per}$, we have
\begin{equation}
\label{eq:t-to-t+1-learn}
\prior^{\per} \neq \prior^{\per+1}
\end{equation}
if $\Demand^{\per}=\demand$.
Notice that  if  $\prior^{\per}\neq\prior^{\infty}$, the set $\learnd(\prior^{\per})$ is  nonempty.

Whenever $\prior^{\per} \neq \dirac_{\State}$, there exist $\state_{1},\state_{2}$ such that
\begin{equation}
\label{eq:mu-theta-1-theta-2}
0 < \prior^{\per}(\state_{1}) < 1 \quad\text{and}\quad
0 < \prior^{\per}(\state_{2}) < 1.
\end{equation}
\cref{eq:mu-theta-1-theta-2} implies that there exists an edge $\edge\in\edges$ such that, for the above states $\state_{1},\state_{2}$,
\begin{equation}
\label{eq:cost-theta-1-cost-theta-2}
\cost_{\edge}(\argdot,\state_{1}) \neq \cost_{\edge}(\argdot,\state_{2}).
\end{equation}
Let $\bar{\load}_{\edge}$ be such that 
\begin{equation}
\label{eq:xbar-cost-theta-1-cost-theta-2}
\cost_{\edge}(\bar{\load}_{\edge},\state_{1}) \neq \cost_{\edge}(\bar{\load}_{\edge},\state_{2}).
\end{equation} 
We want to prove that there exists a value $\demand$ of the demand for which the equilibrium load on edge $\edge$ is $\bar{\load}_{\edge}$.
To do this we use the following lemmata:

\begin{lemma}[\protect{\citet[Proposition~3.12]{ComDosSca:MPB2021}}]
\label{le:increasing-WE-load}
Let $\network$ be a finite \ac{SP} network. 
Then there exists an equilibrium load profile function $\demand\mapsto\eq{\loadprof}(\demand)$ whose components $\eq{\load}_{\edge}(\demand)$ are nondecreasing functions of the demand $\demand$. 
\end{lemma}

The equilibrium edge costs are continuous in the demand \citep[see, e.g.,][Proposition~3.1]{ComDosSca:MPB2021}.
Unboundedness, continuity and monotonicity of the equilibrium edge costs imply 
the following lemma.

\begin{lemma}
\label{le:unbounded-loads}
In a \acl{NRG} played on a SP network, for every $\edge\in\edges$, the equilibrium load map $\eq{\load}_{\edge}$ is unbounded.
\end{lemma}

\begin{proof}[Proof of \cref{le:unbounded-loads}]
By \cref{le:increasing-WE-load} and strict monotonicity of the cost functions, the equilibrium load profile $\eq{\load}_{\edge}$ is weakly increasing.   
Since cost functions are unbounded, continuous and monotonic, equilibrium costs are unbounded as the demand tends to infinity. 
Therefore, for large enough demand, all routes are used and have the same equilibrium cost, which is also unbounded. 
It follows that equilibrium flows on routes are unbounded. 
Therefore equilibrium loads on edges are unbounded as well.
\end{proof}

Continuity of $\cost_{\edge}(\argdot,\state)$ for every $\state\in\states$ and \cref{eq:xbar-cost-theta-1-cost-theta-2} imply that $\cost_{\edge}(\load_{\edge},\state_{1}) \neq \cost_{\edge}(\load_{\edge},\state_{2})$ for every $\load_{\edge}$ in some neighborhood $\interval$ of $\bar{\load}_{\edge}$. 
Moreover, there exists a demand interval $\demandint$ such that, for every $\demand^{\per}\in\demandint$,  the expected equilibrium cost of edge $\edge$ is $\cost_{\edge}(\eq{\load}_{\edge},\prior^{\per})$, with $\eq{\load}_{\edge} \in \interval$.
Therefore, when a demand $\demand^{\per}\in\demandint$ occurs, learning is achieved because one of the two costs $\cost_{\edge}(\eq{\load}_{\edge},\state_{1})$ or $\cost_{\edge}(\eq{\load}_{\edge},\state_{2})$ is realized, so that either 
\begin{equation}
\label{eq:prior-t}
\prior^{\per}(\state_{1})=0, \quad\text{or}\quad 
\prior^{\per}(\state_{2})=0.
\end{equation}

\ref{it:th:learning-strong}
The assumption that  $\supp(\Demand)=[0,\capac)$ implies that the event $\Demand^{\per}\in\demandint$ has positive probability.
Therefore, 
\begin{equation}
\label{eq:P-as-D-t}
\prob\parens*{\Demand^{\per}\in\demandint,\text{ for some }\per\in\naturals}=1,
\end{equation}
which concludes the proof.

\ref{it:th:learning-weak}
If $\prior^{\infty}=\dirac_{\State}$, then strong learning is achieved.
This implies that weak learning is achieved.
If $\prior^{\infty}\neq\dirac_{\State}$, then there exist $\state_{1},\state_{2}\in\states$ and $\bar{\per}$ such that 
\begin{equation}
\label{eq:prior-infty-1-2}
0 < \prior^{\per}(\state_{1}) < 1 \quad\text{and}\quad
0 < \prior^{\per}(\state_{2}) < 1, \quad\text{for all }\per\ge\bar{\per}.
\end{equation}
From the previous arguments, for any such $\state_{1},\state_{2}$, if a value $\demand$ of the demand is such that, for $\history^{\per}=\parens*{\incomp{\history}^{\per-1},\demand}$, either 
\begin{equation}
\label{eq:mu-t-d}
\prior^{\per}(\state_{1})=0\quad\text{or}\quad
\prior^{\per}(\state_{2})=0,
\end{equation}
then  $\demand \notin \supp(\Demand)$.
This shows that  there is  no value in the support of $\Demand$ such that an edge with unknown cost is used. Therefore the split of the flow is the same as it would be under perfect information.
\end{proof}

The following lemma will be used in the proof of \cref{th:converse}.

\begin{lemma}[\protect{\citet[Theorem~3.5]{CheDiaHu:TCS2016}}]
\label{le:Chen}
If a network $\network$ is not \ac{SP}, then it contains an $\source$-$\sink$ paradox, i.e., a subgraph $\graph=\route_{1}\cup\widetilde\route_{2}\cup\widetilde\route_{3}$ such that 
\begin{enumerate}[\upshape(i)]
\item
$\route_{1}$ is a path from $\source$ to $\sink$ that meets in this order the vertices $\source,\nodeA,\vertexalt,\vertex,\nodeB,\sink$,

\item
$\widetilde\route_{2}$ is a path from $\nodeA$ to $\vertex$ whose only vertices in common with $\route_{1}$ are  $\nodeA$ and $\vertex$,

\item
$\widetilde\route_{3}$ is a path from $\vertexalt$ to $\nodeB$ whose only vertices in common with $\route_{1}$ are $\vertexalt$ and $\nodeB$,

\item
$\widetilde\route_{2}$ and $\widetilde\route_{3}$ have no common vertices.
\end{enumerate}
\end{lemma}

\begin{figure}[h]
\newlength{\unit}
\setlength{\unit}{.65cm}
\centering
\begin{tikzpicture}[thick]
	\tikzset{SourceNode/.style = {draw, circle, fill=green!20}}
	\tikzset{AnNode/.style = {draw, rectangle, fill=yellow!20}}
	\tikzset{MidNode/.style = {draw, circle, fill=blue!10}}
	\tikzset{DestNode/.style = {draw, circle, fill=green!20}}
	\node[SourceNode](orig) at (-14\unit,0) {$\scriptstyle{\source}$};
    \node[AnNode](nn1) at (-12\unit,0) {};
    \node[MidNode](a) at (-10\unit,0) {$\scriptstyle{\nodeA}$};
    \node[MidNode](c) at (-8\unit,0) {$\scriptstyle{}$};
    \node[AnNode](nn2) at (-6\unit,0) {};
    \node[MidNode](u) at (-4\unit,0) {$\scriptstyle{\vertexalt}$};
    \node[AnNode](up) at (-4\unit,3\unit) {};
    \node[MidNode](d) at (-2\unit,0) {$\scriptstyle{}$};
	\node[AnNode](nn3) at (0\unit,0) {};
    \node[MidNode](v) at (2\unit,0) {$\scriptstyle{\vertex}$};
    \node[AnNode](down) at (2\unit,-3\unit) {};
    \node[AnNode](nn4) at (4\unit,0) {};
    \node[MidNode](e) at (6\unit,0) {$\scriptstyle{}$};
    \node[MidNode](b) at (8\unit,0) {$\scriptstyle{\nodeB}$};
    \node[AnNode](nn5) at (10\unit,0) {};
	\node[DestNode](dest) at (12\unit,0){$\scriptstyle{\sink}$};
    \draw[->](orig) to (nn1);
    \draw[->](nn1) to (a);
    \draw[->](a) to node[below]{$\edge_{3}$} (c);
    \draw[->](c) to (nn2);
    \draw[dashed,->](a) to  [bend left]  (up);
    \draw[->](nn2) to (u);
    \draw[dashed,->](up) to  [bend left] (v);
    \draw[->](u) to node[above]{$\edge_{1}$} (d);
    \draw[->](d) to   (nn3);
    \draw[dotted,->](u) to [bend right] (down);
    \draw[->](nn3) to (v);
    \draw[->](v) to (nn4);
    \draw[dotted,->](down) to [bend right] (b);
    \draw[->](nn4) to  (e);
    \draw[->](e) to node[above]{$\edge_{2}$}  (b);
    \draw[->](b) to (nn5);
    \draw[->](nn5) to (dest);
		
	\node(L1) at (6\unit,4\unit) {$\route_{1} \colon \source$};		
	\node(L2) at (6\unit,3\unit) {$\compl{\route}_{2} \colon \nodeA$};		
	\node(L3) at (6\unit,2\unit) {$\compl{\route}_{3} \colon \vertexalt$};		
	\node(E1) at (9\unit,4\unit) {$\sink$};		
	\node(E2) at (9\unit,3\unit) {$\vertex$};		
	\node(E3) at (9\unit,2\unit) {$\nodeB$};		
    \draw[->](L1) to (E1);
    \draw[dashed,->](L2) to (E2);
    \draw[dotted,->](L3) to (E3);

\end{tikzpicture}
\caption{\label{fi:paradox} $\source$-$\sink$ paradox. The yellow squares represent finite sequences of nodes connected in series.}
\end{figure}
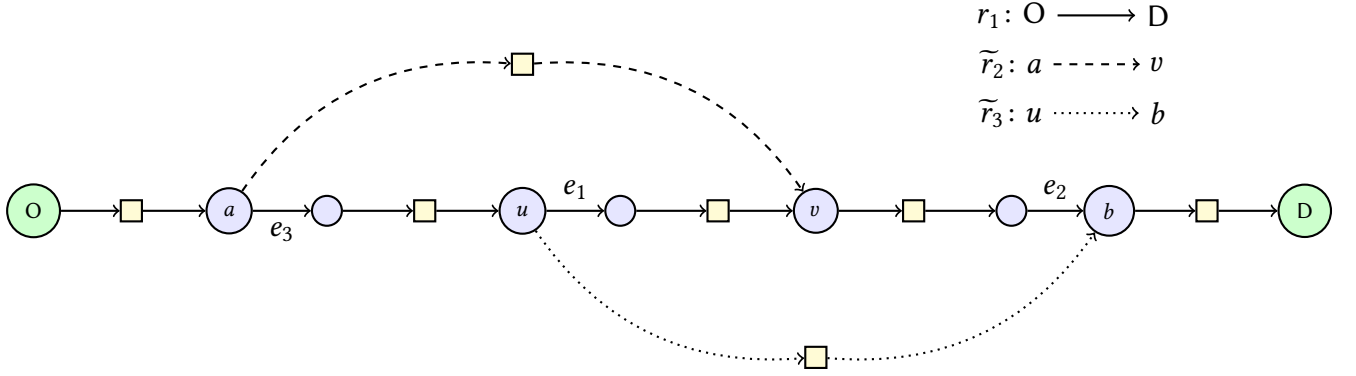

\begin{proof}[Proof of \cref{th:converse}]
Let $\network$ be a network that is not \ac{SP}.
Then, by \cref{le:Chen}, it contains an $\source$-$\sink$ paradox $\network'$ as in \cref{fi:paradox}. 
The idea of the construction is to assign capacities and cost functions to edges in such a way that, the cost function  $\cost_{5}$ is never learned, whatever the feasible demand.

Call $\route_{2}$ the path that coincides with $\route_{1}$ from $\source$ to $\nodeA$, with $\compl{\route_{2}}$ from $\nodeA$ to $\vertex$, and with $\route_{1}$ from $\vertex$ to $\sink$.
Call $\route_{3}$ the path that coincides with $\route_{1}$ from $\source$ to $\vertexalt$, with $\compl{\route_{3}}$ from $\vertexalt$ to $\nodeB$, and with $\route_{1}$ from $\nodeB$ to $\sink$.

Let 
\begin{align*}
\noedges_{\source\nodeA}&= \text{number of edges on $\route_{1}$ between $\source$ and $\nodeA$,}\\
\noedges_{\nodeA\vertexalt}&= \text{number of edges on $\route_{1}$ between $\nodeA$ and $\vertexalt$,}\\
\noedges_{\vertexalt\vertex}&= \text{number of edges on $\route_{1}$ between $\vertexalt$ and $\vertex$,}\\
\noedges_{\vertex\nodeB}&= \text{number of edges on $\route_{1}$ between $\vertex$ and $\nodeB$,}\\
\noedges_{\nodeB\sink}&= \text{number of edges on $\route_{1}$ between  $\nodeB$ and $\sink$,}\\
\noedges_{2}&= \text{number of edges on $\compl{\route}_{2}$,}\\
\noedges_{3}&= \text{number of edges on $\compl{\route}_{3}$.}
\end{align*}

Let $\bigconst$ and $\consteps$ be two positive real numbers such that
\begin{equation}
\label{eq:epsilon}
\consteps < \frac{1}{3} \quad\text{and}\quad 3 < \bigconst.
\end{equation}

All the edges that appear in \cref{fi:paradox} have infinite capacity.
Let $\states=\braces*{\state_{\good},\state_{\bad}}$ with $\prior(\state_{\good})=\prior(\state_{\bad})=1/2$ and let the costs on the edges of the network be as follows:
\begin{align}
\label{eq:cost-paradox-2}
\cost_{2}(\load,\state)&=\parens*{\bigconst+\frac{\consteps}{\noedges_{\vertex\nodeB}}}\load, \quad\text{ for all }\state\in\states\\
\label{eq:cost-paradox-3}
\cost_{3}(\load,\state)&=\parens*{\bigconst+\frac{\consteps}{\noedges_{\nodeA\vertexalt}}}\load, \quad\text{ for all }\state\in\states\\
\cost_{1}(\load,\state)&=
\begin{cases}
\parens*{\bigconst+\dfrac{\consteps}{\noedges_{\vertexalt\vertex}}}\load &\text{ for } \load \le 1, \quad\text{ for all }\state\in\states\\
\parens*{\bigconst+\dfrac{\consteps}{\noedges_{\vertexalt\vertex}}}+\consteps^{2}\parens*{\load-1} &\text{ for }\load >1 \text{ and } \state=\state_{\good},\\
\parens*{\bigconst+\dfrac{\consteps}{\noedges_{\vertexalt\vertex}}}+\parens*{2\bigconst+\dfrac{2\consteps}{\noedges_{\vertexalt\vertex}}-\consteps^{2}}\parens*{\load-1} &\text{ for }\load >1 \text{ and } \state=\state_{\bad}.
\end{cases}
\end{align}
For every other edge $\edge$ appearing in \cref{fi:paradox} the cost functions are as follows:
\begin{equation}
\label{eq:edges-in-figure}
\begin{split}
\cost_{\edge}(\load,\state)&=\frac{\consteps}{\noedges_{\source\nodeA}}\load, \quad\text{ for all }\state\in\states, \quad\text{if $\edge$ is between $\source$ and $\nodeA$},\\
\cost_{\edge}(\load,\state)&=\frac{\consteps}{\noedges_{\nodeA\vertexalt}}\load, \quad\text{ for all }\state\in\states, \quad\text{if $\edge$ is between $\nodeA$ and $\vertexalt$},\\
\cost_{\edge}(\load,\state)&=\frac{\consteps}{\noedges_{\vertexalt\vertex}}\load, \quad\text{ for all }\state\in\states, \quad\text{if $\edge$ is between $\vertexalt$ and $\vertex$},\\
\cost_{\edge}(\load,\state)&=\frac{\consteps}{\noedges_{\vertex\nodeB}}\load, \quad\text{ for all }\state\in\states, \quad\text{if $\edge$ is between $\vertex$ and $\nodeB$},\\
\cost_{\edge}(\load,\state)&=\frac{\consteps}{\noedges_{\nodeB\sink}}\load, \quad\text{ for all }\state\in\states, \quad\text{if $\edge$ is between $\nodeB$ and $\sink$},\\
\cost_{\edge}(\load,\state)&=\frac{\consteps}{\noedges_{2}}\load, \quad\text{ for all }\state\in\states, \quad\text{if $\edge$ is on $\compl{\route}_{2}$},\\
\cost_{\edge}(\load,\state)&=\frac{\consteps}{\noedges_{3}}\load, \quad\text{ for all }\state\in\states, \quad\text{if $\edge$ is on $\compl{\route}_{3}.$}
\end{split}
\end{equation}
All the edges $\edge$ that do not appear in \cref{fi:paradox} have a capacity 
\begin{equation}
\label{eq:capacity-outside}
\capac_{\edge} = \frac{\capkappa}{\abs{\routes}},
\end{equation}
where $\abs{\routes}$ is the cardinality of $\routes$ and 
\begin{equation}
\label{eq:kappa}
\capkappa<\frac{1}{2}.
\end{equation}
Moreover, for these edges
\begin{equation}
\label{eq:cost-outside-figure}
\cost_{\edge}(\load_{\edge}) = \frac{1}{\capac_{\edge}-\load_{\edge}}, \quad\text{for  }\load_{\edge}\in[0,\capac_{\edge}).
\end{equation}

\cref{eq:capacity-outside} implies that the  load on edge $\edge_{1}$ coming from flows of paths different from $\route_{1}$ is smaller than $1$.

We   prove now that, in equilibrium, the total load on edge $\edge_{1}$ is smaller than $1$. 
Let $\flowprof$ be a feasible flow vector and let $\flow_{1},\flow_{2},\flow_{3}$ be the corresponding flows on $\route_{1},\route_{2},\route_{3}$, respectively.
The expected costs given $\prior$ satisfy the following inequalities
\begin{equation}
\label{eq:expected_costs}
\begin{split}
\cost_{\route_1}(\flowprof,\prior)
& \ge \consteps(\flow_{1}+\flow_{2}+\flow_{3})
+ (\consteps+\bigconst)(\flow_{1}+\flow_{3})
+ (\consteps+\bigconst)\flow_{1}
+ (\consteps+\bigconst)(\flow_{1}+\flow_{2})
\\
\cost_{\route_{2}}(\flowprof,\prior)
& \le \consteps(\flow_{1}+\flow_{2}+\flow_{3}+\capkappa)
+\consteps(\flow_{2}+1)
+ (\consteps+\bigconst)(\flow_{1}+\flow_{2}+\capkappa) 
\\
\cost_{\route_{3}}(\flowprof,\prior)
& \le \consteps(\flow_{1}+\flow_{2}+\flow_{3}+\capkappa)
+ (\consteps+\bigconst)(\flow_{1}+\flow_{3}+\capkappa)
+ \consteps(\flow_{3}+\capkappa)
\end{split}
\end{equation}

The path $\route_1$ has a positive flow in equilibrium if and only if 
\begin{align}
\label{eq:condit-1-2}
\cost_{\route_1}(\flowprof,\prior) 
&\le \cost_{\route_{2}}(\flowprof,\prior)
\intertext{and}
\label{eq:condit-1-3}
\cost_{\route_1}(\flowprof,\prior) 
&\le \cost_{\route_{3}}(\flowprof,\prior).
\end{align}
The inequalities in \cref{eq:expected_costs,eq:condit-1-2} imply
\begin{equation}
\label{eq:condit-alt-2}
(\consteps+\bigconst)(2\flow_{1}+\flow_{3})
\le \consteps(\flow_{2}+\capkappa) + (2\consteps+\bigconst)\capkappa.
\end{equation}
Similarly, from  \cref{eq:expected_costs,eq:condit-1-3}, we obtain
\begin{equation}
\label{eq:condit-alt-3}
(\consteps+\bigconst)(2\flow_{1}+\flow_{2})
\le \consteps(\flow_{3}+\capkappa) + (2\consteps+\bigconst)\capkappa.
\end{equation}
Summing \cref{eq:condit-alt-2,eq:condit-alt-3}, we obtain
\begin{equation}
\label{eq:cond-alt-4}
(\consteps+\bigconst)(4\flow_{1}+\flow_{3}+\flow_{2})
\le
(2\bigconst + 6\consteps)\capkappa + \consteps(\flow_{2}+\flow_{3}),
\end{equation}
that is,
\begin{equation}
\label{eq:cond-alt-}
\flow_{1}
\le
\frac{(2\bigconst + 6\consteps)\capkappa - \bigconst (\flow_{3}+\flow_{2})}{4(\consteps+\bigconst)}
\le
\frac{(\bigconst+3\consteps)\capkappa}{2(\consteps+\bigconst)}
\le
\frac{(\bigconst+3\consteps)\capkappa}{2\bigconst}
\le \capkappa.
\end{equation}
Therefore, because of \cref{eq:kappa}, the load on $\edge_{1}$ is at most $\capkappa+\capkappa\le 1$.
This implies that the cost function $\cost_{1}$ is not learned,  for any value of the demand.
On the other hand, if the true state were known to be $\state_{\good}$, the equilibrium flow would be different than the one under uncertainty. 
This shows that weak learning is not achieved. 
\end{proof}

\section{Concluding remarks}\label{se:signals}

\subsection{Random demand vs noisy costs}

In our model, the demand is random and realized costs are observed with certainty.  \citet{WuAmi:IFAC2019} also considered dynamic nonatomic routing games with unknown states where realized costs are observed. 
Unlike what we do in our paper, they assumed constant demand and noisy costs, that is, in their model realized costs depend on  the unknown state and on  multivariate normally distributed noises. The following  example shows that these different sources of randomness lead to different learning outcomes.

Consider the network in \cref{fi:parallel-network} with infinite capacities and
\begin{equation}
\label{eq:costs-Pigou-3}
\cost_{1}(\load,\state)=\load,\qquad
\cost_{2}(\load,\state)=
\begin{cases}
\load &\text{for }\state=\state_{\good},\\
\load + \consta &\text{for }\state=\state_{\bad},
\end{cases}
\end{equation}
with $\consta>0$ and  $\prior(\state_{\bad})=\prior(\state_{\good})=1/2$.
As shown before, the expected cost of edge $\edge_{2}$ is 
\begin{equation}
\label{eq:expected-cost-e2+}
\cost_{2}(\load,\prior) = \load + \frac{1}{2}\consta.
\end{equation}
Assume that $\Demand^{\per}$ has an exponential distribution with mean $\consta/4$. 
As $\supp(\Demand^{\per})=\reals_{+}$, we have
\begin{equation}
\label{eq:PDt>a-as}
\prob\parens*{\Demand^{\per}>\consta/2,\text{ for some }\per\ge 1}=1,
\end{equation}
which implies that, almost surely, edge $\edge_{2}$ is used at some point and, in our model, strong learning occurs. 

Consider now the same network game with the information model of \citet{WuAmi:IFAC2019} with the demand $\demand^{\per}=\consta/4$, i.e.,  equal to the expected demand of our model. 
Let the observed costs at period $\per$ be realizations of the random variables
\begin{equation}
\label{eq:noisy-costs}
\noisy{\cost}_{\edge}^{\per}(\load) \coloneqq\cost_{\edge}(\load)+\consteps^{\per},
\end{equation}
for $\edge=1,2$, with $\consteps^{\per}$ normally distributed with mean $0$ and  variance $\sigma^{2}$.
No matter  the realization of the random cost, at any period $\per$, the expected cost of edge $1$ is lower than the expected cost of edge $2$, so edge $2$ is never used and  weak learning fails.
In this example, there is learning only when high demand forces exploration of the edge with unknown cost. This shows that, despite their similarities, from the perspective of learning, the model with random demand and the one with noisy costs are  different.

\subsection{Future work}

Although the previous example shows that  models with noisy observation of realized costs and  models with variable demand yield different properties, an interesting open question is how these two sources of randomness interact when combined. One particular case of interest is a model where the variance of the noise on a given edge is proportional to either its equilibrium load or the total demand. 

Another direction  of extension  is to provide bounds on the speed of learning   for specific classes of cost functions. Little can be said on convergence speed without restrictions on cost functions. Indeed, two functions may  differ  by an  arbitrarily small margin, or on a set of  arbitrarily small probability. In line with the literature, restricting to a specific class of costs  may yield computable bounds.

\section{List of symbols}
\begin{longtable}{p{.13\textwidth} p{.82\textwidth}}

$\nodeA$ & vertex\\
$\nodeB$ & vertex\\
$\Bor$ & Borel $\sigma$-field\\
$\cost_{\edge}$ & cost of edge $\edge$\\
$\cost_{\edge}(\load,\prior)$ & expected cost of edge $\edge$ when the prior is $\prior$, defined in \cref{eq:cost-prior}\\
$\cost_{\route}$ & cost of path $\route$\\
$\cost_{\route}(\flowprof,\prior)$ & expected cost of path $\route$ when the prior is $\prior$, defined in \cref{eq:cost-prior}\\
$\costprof(\loadprof)$ & cost vector induced by the load vector $\loadprof$\\
$\cut$ & cut\\
$\demand$ & traffic demand\\
$\Demand^{\per}$ & random traffic demand at period $\per$\\
$\sink$ & destination\\
$\edge$ & edge\\
$\edges$ & set of edges\\
$\distr$ & demand distribution\\
$\nagame$ & nonatomic routing game \\
$\nagame_{\prior}$ & nonatomic routing game with unknown  states\\
$\history^{\per}$ & history at period $\per$\\
$\network$ & oriented multigraph\\
$\source$ & origin\\
$\prob$ & product measure $\prior\otimes\distr^{\infty}$\\
$\route$ & feasible path\\
$\routes$ & set of feasible paths from $\source$ to $\sink$\\
$\supp$ & support of a random variable\\
$\states$ & state space\\
$\vertexalt$ & vertex\\
$\vertex$ & vertex\\
$\vertices$ & set of vertices\\
$\load_{\edge}$ & load of edge $\edge$\\
$\loadprof$ & load vector\\
$\eq{\loadprof}$ & equilibrium load vector\\
$\flow_{\route}$ & flow of path $\route$\\
$\flowprof$ & flow vector\\
$\eq{\flowprof}$ & equilibrium flow vector\\
$\flows$ & set of feasible flows\\

$\convalpha$ & scalar in $[0,1]$\\
$\capac_{\edge}$ & capacity of edge $\edge$\\
$\capac_{\cut}$ & capacity of cut $\cut$\\
$\capac$ & capacity of the network\\
$\game$ & dynamic nonatomic routing game with unknown states\\
$\dirac_{\state}$ & Dirac measure on $\state$\\
$\simplex(\states)$ & simplex of probability measures on $\states$\\
$\State$ & random state of nature\\
$\state$ & possible value of state of nature\\
$\prior$ & prior distribution on $\states$\\
$\prior^{\per}$ & posterior distribution\\
$\prior^{\infty}$ & almost sure limit of $\prior^{\per}$\\

\end{longtable}

\subsection*{Acknowledgments}

Marco Scarsini is a member of GNAMPA-INdAM. 
His research received partial support from the COST action GAMENET, the INdAM-GNAMPA Project 2020 ``Random walks on random games,'' and  the Italian MIUR PRIN 2017 Project ALGADIMAR ``Algorithms, Games, and Digital Markets.'' 
Tristan Tomala gratefully acknowledges the support of the HEC Foundation and ANR/Investissements d'Avenir under grant ANR-11-IDEX-0003/Labex Ecodec/ANR-11-LABX-0047.
Emilien Macault gratefully acknowledges the support of HEC Foundation and the COST action GAMENET.

{
\hypersetup{linkcolor=red}
\bibliographystyle{apalike}
\bibliography{bibbrg}
}

\newpage

\end{document}